\documentclass[conference,letterpaper]{IEEEtran}

\usepackage[utf8]{inputenc} 
\usepackage[T1]{fontenc}
\usepackage{url}
\usepackage{ifthen}
\usepackage{cite}
\usepackage[cmex10]{amsmath} 

\usepackage{amssymb}
\usepackage{stmaryrd}
\usepackage{epic,eepic}
\usepackage{theorem}
\usepackage{pifont}  
\usepackage{euscript}
\usepackage{calc}
\usepackage[titles]{tocloft}
\usepackage[usenames]{color}          
\usepackage{xcolor}

\usepackage{shortvrb,psfrag}
\usepackage{graphicx}       
\usepackage{dsfont}
\usepackage{tikz}
\usepackage{pgfplots}

\newtheorem{lemma}{Lemma}

\newtheorem{theorem}{Theorem}
\newtheorem{definition}{Definition}

\newcommand{\mc}{\mathcal}

\newcommand{\PP}{\mathcal{P}}
\newcommand{\QQ}{\mathcal{Q}}
\newcommand{\E}{\mathbb{E}}
\newcommand{\R}{\mathbb{R}}
\newcommand{\N}{\mathbb{N}}

\newcommand{\supp}{\mathrm{supp}\,}

\newcommand{\argmin}{\ensuremath{\operatorname{argmin}}}

\usepackage{lipsum}

\newcommand\blfootnote[1]{%
  \begingroup
  \renewcommand\thefootnote{}\footnote{#1}%
  \addtocounter{footnote}{-1}%
  \endgroup
}

\begin{document}
\title{Single-Letter Characterization of the Mismatched Distortion-Rate Function}


\author{%
  \IEEEauthorblockN{Ma\"{e}l Le~Treust}
  \IEEEauthorblockA{Univ. Rennes, CNRS, Inria, IRISA UMR 6074\\
                    35000 Rennes, France\\
                    Email: mael.le-treust@cnrs.fr}     
  \and
  \IEEEauthorblockN{Tristan Tomala}
  \IEEEauthorblockA{HEC Paris\\ 
                    78351 Jouy-en-Josas, France\\
                    Email: tomala@hec.fr}
}

\maketitle

\blfootnote{The work of Ma\"{e}l Le Treust is supported in part by PEPR NF FOUNDS ANR-22-PEFT-0010. Tristan Tomala gratefully acknowledges the support of the HEC foundation and ANR/Investissements d'Avenir under grant ANR-11-IDEX-0003/Labex Ecodec/ANR-11-LABX-0047.}  


\begin{abstract}
The mismatched distortion-rate problem has remained open since its formulation by Lapidoth in 1997. In this paper, we characterize the mismatched distortion-rate function. Our single-letter solution highlights the adequate conditional distributions for the encoder and the decoder. The achievability result relies on a time-sharing argument that allows to convexify the upper bound of Lapidoth. We show that it is sufficient to consider two regimes, one with a large rate and another one with a small rate. Our main contribution is the converse proof. Suppose that the encoder selects a single-letter conditional distribution distinct from the one in the solution, we construct an encoding strategy that leads to the same expected cost for both encoder and decoder. This ensures that the encoder cannot gain by changing the single-letter conditional distribution. This argument relies on a careful identification of the sequence of auxiliary random variables. By building on Caratheodory's Theorem we show that the cardinality of the auxiliary random variables is equal to the one of the source alphabet plus three.

\end{abstract}

\section{Introduction}

In the mismatched distortion-rate problem, the encoder optimizes a distortion function which differs from the one of the decoder. Formulated by Lapidoth in \cite{Lapidoth97}, the main difficulty comes from the fact that the decisions are taken separately. This problem differs from Shannon's distortion-rate problem \cite{Shannon59} where the distortion function may have several components, and both  encoding and decoding functions are selected jointly. Here, the encoder and the decoder decide in a decentralized manner which coding functions to implement, and they optimize distinct cost functions. 
The use of Game Theory is thus relevant. It gives us new insights to solve the mismatched distortion-rate problem. 


By selecting in advance the decoding function, the decoder  determines the codebook. Indeed, the decoder can {remove} a codeword simply by producing the same decoder output as for another codeword. The encoder knows the decoding function before the transmission starts. Once it observes the source sequence, it selects the optimal codeword for its own cost function. This reaction is anticipated by the decoder 
when selecting the decoding function. This defines a Stackelberg game \cite{stackelberg-book-2011} in which the decoder is the leader that plays first, and the encoder is the follower that plays second \cite{VoraKulkarni_ISIT20}. 

When there is no communication constraint, this problem coincides with the mechanism design problem solved by Jackson and Sonnenschein in \cite{JacksonSonnenschein07}. A large communication rate clearly benefits to the decoder. According to the \emph{revelation principle}, see \cite[Chap.~7.2]{tirole-book-1991}, it is optimal that the decoding scheme incentivizes the encoder to communicate perfectly
the source sequence to the decoder. However, with restrictions in the communication, the source cannot be perfectly transmitted to the decoder, and thus the revelation principle breaks down. 

\vspace{-0.45cm}

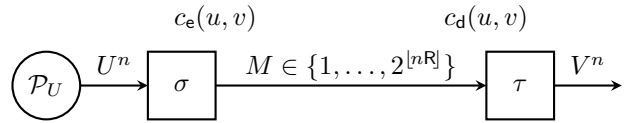
\begin{figure}[ht!]
\begin{center}
\begin{tikzpicture}[xscale=0.9,yscale=0.9]
\draw [thick] (0,0.5) circle (0.5) ;
\draw [thick,>=stealth,->] (0.5,0.5) -- (1.5,0.5) ;
\draw [thick] (1.5,0) rectangle (2.5,1);
\draw [thick,>=stealth,->] (2.5,0.5) -- (6.5,0.5) ;
\draw [thick] (6.5,0) rectangle (7.5,1);
\draw [thick,>=stealth,->] (7.5,0.5) -- (8.5,0.5) ;
\draw (0,0.5) node {$\PP_U$} ;
\draw (2,0.5) node {$\sigma$} ;
\draw (7,0.5) node {$\tau$} ;
\draw  (1,0.8) node {$U^n$} ;
\draw  (4.5,0.8) node {$M\in\{1,\ldots,2^{\lfloor \!n \mathsf{R}\!\rfloor}\}$} ;
\draw   (8,0.8) node {$V^n$} ;
\draw   (2.5,1.5) node {$c_{\mathsf{e}}(u,v)$} ;
\draw   (6.5,1.5) node {$c_{\mathsf{d}}(u,v)$} ;
\end{tikzpicture}
\end{center}
\vspace{-0.35cm}
\caption{The mismatched distortion-rate problem.}\label{fig:Problem}
\end{figure}

\vspace{-0.25cm}

This problem is fundamental to understand situations in which the decoding rule is selected exogenously, for example in learning algorithms. Kanabar and Scarlett use superposition coding and expurgated parallel coding to derive new achievability bounds in \cite{KanabarScarlett24}. 
The authors of \cite{ScarlettGuillenSomekhMartinez_FTCIT20} pointed out a list of open questions covering the mismatched distortion-rate problem and the mismatched channel coding problem, for which converse bounds are provided in \cite{Somekh22} and \cite{KangarshahiFabregas21}. In \cite{LeTreustTomala_ISIT24}, we provide a single-letter converse bound that matches Lapidoth's achievability bound in several cases. 

In this paper, we characterize the single-letter mismatched distortion-rate function. The achievability proof relies on the time-sharing argument that convexifies the achievability bound of Lapidoth. We show that it is sufficient to consider two regimes, one with a large rate and another one with a small rate. The main novelty lies in the converse proof. For each component of the source and decoder output symbols $(U_t,V_t)$ with $t\in\{1,\ldots,n\}$, we identify  the auxiliary random variable $W_t  = (M,U^{t-1})$. Suppose that the encoder selects a single-letter conditional distribution distinct from the one in the solution, we construct an encoding strategy that leads to the same expected cost for the encoder and decoder. Therefore, the encoder cannot gain by changing the single-letter conditional distribution. We build on Caratheodory's Theorem  to bound the cardinality of each of the auxiliary random variables $(W_t)_{t}$ by the cardinality of the source alphabet plus three.




\section{System Model} 

A memoryless source takes values in the finite set $\mc{U}$, according to the probability distribution $\PP_U\in\Delta(\mc{U})$. The notation $\Delta(\cdot)$ stands for the set of probability distributions. The encoder and decoder optimize two distinct cost functions
\begin{align}
c_{\mathsf{e}}: \mc{U}\times \mc{V} \rightarrow \R,\qquad
c_{\mathsf{d}}: \mc{U}\times \mc{V} \rightarrow \R,\label{eq:distortion}
\end{align}
where $\mc{V}$ denotes the finite set of decoder output symbols.

\begin{definition}\label{def:Code}
Given $\mathsf{R}\geq0$ and $n\in \N^{\star}$, the stochastic encoding and the decoding functions are defined by
\begin{align}
\sigma:&\;\,\mc{U}^{n} \longrightarrow \Delta\big(\{1,\ldots,2^{\lfloor \!n \mathsf{R}\!\rfloor}\}\big),\\
\tau :&\;\, \{1,\ldots,2^{\lfloor \!n \mathsf{R}\!\rfloor}\}   \longrightarrow  \Delta( \mc{V}^n).
\end{align}
\end{definition}

The pair $(\sigma,\tau)$ induces the probability distribution $\PP_{U^nMV^n}^{\sigma,\tau}$. We denote the encoder long-run cost function by
\begin{align}
c^n_{\mathsf{e}}(\sigma,\tau) = \sum_{u^n,v^n}\PP_{U^nV^n}^{\sigma,\tau}\big(u^n,v^n \big) \cdot  \Bigg[  \frac{1}{n} \sum_{t=1}^n c_{\mathsf{e}}(u_t,v_t) \Bigg].
\end{align}
Similarly, the decoder long-run cost function writes $c^n_{\mathsf{d}}(\sigma,\tau)$.

\begin{definition}\label{def:EncBR} 
Given $\mathsf{R}\geq0$ and $n\in \N^{\star}$, we define\\
1.  the set of encoder best responses to the decoding $\tau$ 
\begin{align}
\textsf{BR}_{\mathsf{e}}(\tau) =&\underset{\sigma}{\argmin}\; c_{\mathsf{e}}^{n}(\sigma, \tau),
\end{align}
2. the $n$-optimal decoder cost with pessimistic tie-breaking rule
\begin{align}
C_{\mathsf{d}}^n(\mathsf{R})=\inf_{\tau}\max_{\sigma  \in \textsf{BR}_{\mathsf{e}}(\tau)} c_{\mathsf{d}}^{n}(\sigma, \tau). \label{eq:MDProblem}
\end{align}
\end{definition}

The mismatched distortion-rate function is defined by
\begin{align}
C_{\mathsf{d}}^{\star}(\mathsf{R}) = \lim_{n\to+\infty} C_{\mathsf{d}}^n(\mathsf{R})= \inf_{n\in\N^{\star}}C_{\mathsf{d}}^n(\mathsf{R}).
\end{align}
The limit coincides with the infimum according to Fekete's Lemma for the subadditive sequence $(n\cdot C_{\mathsf{d}}^n(\mathsf{R}))_{n\in \N^{\star}}$, see \cite[Lemma~4]{LeTreustTomala_ISIT24}.

\begin{theorem}\label{theo:main}
The mismatched distortion-rate function writes
\begin{align}
C_{\mathsf{d}}^{\star}(\mathsf{R}) = \min_{\alpha\in[0,1], \mathsf{R}_1\geq0, \mathsf{R}_2\geq0,\atop \alpha \mathsf{R}_1 +(1- \alpha) \mathsf{R}_2 \leq \mathsf{R}} \alpha C(\mathsf{R}_1) + (1- \alpha)C(\mathsf{R}_2).
\end{align}
It is the convexification, i.e. the greatest convex minorant, of 
\begin{align}
C(\mathsf{R}) =& \inf_{\QQ_{WV}}\; \; \max_{\QQ_{W|U}\in \mathbb{P}(\QQ_{WV},\mathsf{R})}\; \E\Big[ c_{\mathsf{d}}(U,V) \Big],\label{eq:DefinitionFunctionC}
\end{align}
where $W\in\mc{W}$ is an auxiliary random variable with $|\mc{W}|= |\mc{U}|+3$. The set of \emph{compatible distributions} $ \mathbb{P}(\QQ_{WV},\mathsf{R})$, and the set of \emph{encoder's deviations}  $\mathbb{D}(\QQ_{W},\mathsf{R})$ are defined by
\begin{align}
\mathbb{P}(\QQ_{WV},\mathsf{R}) =&  \quad \underset{\PP_{W|U}\in \mathbb{D}(\QQ_{W},\mathsf{R})}{\argmin} \;\;\, \E_{\QQ}[  c_{\mathsf{e}}(U,V) ], \\
\mathbb{D}(\QQ_{W},\mathsf{R}) =& \;\; \Big\{ \PP_{W|U}\in \Delta( \mc{W})^{|\mc{U}|},
\;\; I_{\PP}(U;W)   \leq  \mathsf{R},\nonumber\\
&\quad\quad  \;\;\sum_u\PP_U(u)\PP_{W|U}(\cdot|u) = \QQ_{W} \Big\}.
\end{align}
\end{theorem}

The achievability proof is obtained by using a time-sharing argument for the coding scheme of \cite[Sec.~II]{Lapidoth97}, see also \cite[Sec.~4.5]{ScarlettGuillenSomekhMartinez_FTCIT20} and \cite[Sec.~III]{LeTreustTomala_ISIT24}. We focus on the converse proof.

\section{Converse Proof of Theorem~\ref{theo:main}}

\subsection{Identification of the auxiliary random variables}

We fix $n \in \N^{\star}$, we consider a decoding function $\tau$ and an encoding function $\sigma \in \textsf{BR}_{\mathsf{e}}(\tau)$ that achieves the maximum in
\begin{align}
\max_{\sigma  \in \textsf{BR}_{\mathsf{e}}(\tau)} c_{\mathsf{d}}^{n}(\sigma, \tau).\label{eq:MaxTie}
\end{align}

Let $\mc{T} =\{1,\ldots,n\}$, for all $t\in\{1,\ldots,n\}$ we identify the coefficients $\alpha_t = \frac1n$,  the auxiliary random variables, the rates
\begin{align}
W_t  =& \; (M, U^{t-1}), \label{eq:DefAuxRV}\\
\mathsf{R}_t  =& \;  I(U_t;W_t) = I(U_t;M,U^{t-1}).\label{eq:RateW}
\end{align}

We denote by $\QQ^{\sigma,\tau}_{U^nW^nV^n} $ the distribution of the sequences $(U^n, W^n,V^n)$ induced by $(\sigma,\tau)$, and we denote by $\QQ^{\sigma,\tau}_{U_tW_tV_t}$ the marginal for all $t$. 
We define the average distribution by
\begin{align}
\QQ^{\sigma,\tau}_{UWV} = \frac1n\sum_{t =1}^n \QQ^{\sigma,\tau}_{U_tW_tV_t}.\label{eq:Distr}
\end{align}

In this proof, we consider the above $\mc{T}$, $(\alpha_t)_{t\in\mc{T}}$, $(\mathsf{R}_t)_{t\in\mc{T}}$, $(U_t,W_t,V_t)_{t\in\mc{T}}$, and we show that $\QQ^{\sigma,\tau}_{UWV} $ belongs to the set
\begin{align}
\mathbb{Q} =  \Bigg\{& \PP_U\QQ_{V|U},\;  \exists \mc{T}, \exists (\alpha_t)_{t\in\mc{T}},  \exists ( \mathsf{R}_t)_{t\in\mc{T}}, \exists (U_t,W_t,V_t)_{t\in\mc{T}} ,\nonumber\\
&\forall t\in \mc{T}, \quad U_t  -\!\!\!\!\minuso\!\!\!\!- W_t  -\!\!\!\!\minuso\!\!\!\!- V_t, \;\;\QQ_{U_t} = \PP_U,\label{eq:SetCond_2}\\
& \forall t\in \mc{T}, \quad \QQ_{W_t|U_t}\in \mathbb{P}(\QQ_{W_tV_t},\mathsf{R}_t),\text{ and }  \nonumber\\
&\quad\E_{\QQ}\Big[ c_{\mathsf{d}}(U_t,V_t) \Big] = \!\!\! \max_{\QQ_{W|U}\in \mathbb{P}(\QQ_{W_tV_t},\mathsf{R}_t)}\E\Big[ c_{\mathsf{d}}(U_t,V_t) \Big],  \label{eq:SetCond_3}\\
&\sum_{t\in\mc{T}}\alpha_t \mathsf{R}_t \leq  \mathsf{R},\qquad  \sum_{t\in\mc{T}}\alpha_t =1,\label{eq:SetCond_4}\\
&\PP_U\QQ_{V|U} = \sum_t \alpha_t \sum_{w_t} \PP_{U_t}\QQ_{W_t|U_t}(w_t|\cdot)\QQ_{V_t|W_t}(\cdot|w_t)\Bigg\}.\nonumber
\end{align}
Lemma~\ref{lemma:MarkovMarginal}, Lemma~\ref{lemma:UnprofitableDeviation}, Lemma~\ref{lemma:rate} below, show that $\QQ^{\sigma,\tau}_{U_tW_tV_t}$ satisfies  \eqref{eq:SetCond_2}, \eqref{eq:SetCond_3}, \eqref{eq:SetCond_4} for all $t\in\mc{T}$. Finally, we have 
\begin{align}
&\max_{\sigma  \in \textsf{BR}_{\mathsf{e}}(\tau)} c_{\mathsf{d}}^{n}(\sigma, \tau) 
= \E_{\QQ^{\sigma,\tau}}\Big[ c_{\mathsf{d}}(U,V) \Big]\label{eq:ConverseFinal1}\\
&= \sum_{t=1}^n \alpha_t\max_{\QQ_{W|U}\in \mathbb{P}(\QQ^{\tau}_{W_tV_t},\mathsf{R}_t)}\E\Big[ c_{\mathsf{d}}(U_t,V_t) \Big]\label{eq:ConverseFinal3}\\
&\geq \sum_{t=1}^n \alpha_t\inf_{\QQ_{WV}}\max_{\QQ_{W|U}\in \mathbb{P}(\QQ_{WV},\mathsf{R}_t)}\E\Big[ c_{\mathsf{d}}(U_t,V_t) \Big]\label{eq:ConverseFinal4}\\
&= \sum_{t=1}^n \alpha_t C(\mathsf{R}_t)\label{eq:ConverseFinal5}\\
&\geq \min_{(\tilde{\alpha}_t)_t\geq0,\atop (\tilde{\mathsf{R}}_t)_t\geq0}\bigg\{\sum_{t=1}^n {\tilde{\alpha}}_t  C(\tilde{\mathsf{R}}_t)\; \bigg|\;  \sum_{t=1}^n \tilde{\alpha}_t = 1,\;  \sum_{t=1}^n\tilde{\alpha}_t \tilde{\mathsf{R}}_t \leq \mathsf{R}\bigg\}.\label{eq:ConverseFinal6}
\end{align}
Equations \eqref{eq:ConverseFinal1}-\eqref{eq:ConverseFinal3} come from the definition of $\QQ^{\sigma,\tau}_{UWV}$ in \eqref{eq:Distr} and the properties of $(\sigma,\tau)$ in \eqref{eq:MaxTie}. 

It remains to prove the cardinality bounds for the auxiliary random variables $|\mc{T}|\leq 2$ and $|\mc{W}|\leq |\mc{U}| +3$, in Lemma~\ref{lemma:CardT} and Lemma~\ref{lemma:CardW}. This conclude the converse proof of Theorem~\ref{theo:main}.

\subsection{Distribution $\QQ^{\sigma,\tau}_{UWV} $ belongs to the set $\mathbb{Q}$}\label{sec:Lemma}

\begin{lemma}\label{lemma:MarkovMarginal}
For all $t\in\{1,\ldots,n\}$,  $\QQ^{\sigma,\tau}_{U_tW_tV_t}$ satisfies
\begin{align}
U_t  -\!\!\!\!\minuso\!\!\!\!- W_t  -\!\!\!\!\minuso\!\!\!\!- V_t \quad\text{ and } \quad \PP_{U_t} = \PP_U.
\end{align}
\end{lemma}

\begin{proof}[Lemma \ref{lemma:MarkovMarginal}]
The system model induces the Markov chain $U^n  -\!\!\!\!\minuso\!\!\!\!- M  -\!\!\!\!\minuso\!\!\!\!- V^n$. Therefore, 
\begin{align}
I(U^n;V^n|M) = 0 \quad \Longrightarrow& \quad I(U_t,U^{t-1};V_t|M) = 0 \\
\quad \Longrightarrow& \quad I(U_t;V_t|M,U^{t-1}) = 0\\
\Longleftrightarrow& \quad I(U_t;V_t|W_t) = 0.
\end{align}
This concludes the proof of Lemma \ref{lemma:MarkovMarginal}.
\end{proof}

The Markov chain induces $\QQ^{\sigma,\tau}_{U_tW_tV_t}= \PP_U \QQ^{\sigma}_{W_t|U_t} \QQ^{\tau}_{V_t|W_t}$.
\begin{lemma}\label{lemma:UnprofitableDeviation}
For all $ t\in\{1,\ldots,n\}$, we have
\begin{align}
&\QQ^{\sigma}_{W_t|U_t}\in \mathbb{P}(\QQ^{\tau}_{W_tV_t},\mathsf{R}_t),\text{ and }\\
&\E_{\QQ^{\sigma,\tau}}\Big[ c_{\mathsf{d}}(U_t,V_t) \Big] = \!\!\! \max_{\QQ_{W|U}\in \mathbb{P}(\QQ^{\tau}_{W_tV_t},\mathsf{R}_t)}\E\Big[ c_{\mathsf{d}}(U_t,V_t) \Big] .
\end{align}
\end{lemma}
The proof of Lemma \ref{lemma:UnprofitableDeviation} is stated in Sec.~\ref{sec:ProofLemma42}.

\begin{lemma}\label{lemma:rate}
For all $t\in\{1,\ldots,n\}$, 
we have
\begin{align}
\sum_{t=1}^n \alpha_t \mathsf{R}_t \leq \mathsf{R}.
\end{align}
\end{lemma}

\begin{proof}[Lemma \ref{lemma:rate}]
Since $U^n$ is i.i.d. distributed, we have
\begin{align}
&\sum_{t=1}^n\alpha_t  \mathsf{R}_t =\sum_{t=1}^n \frac1nI(U_t;W_t) = \frac1n\sum_{t=1}^n I(U_t;M,U^{t-1})\nonumber\\
 =&\frac1n\sum_{t=1}^n I(U_t;M|U^{t-1})  = \frac1nI(U^n;M)\leq \frac1nH(M)\leq  \mathsf{R}.\nonumber
\end{align}
This concludes the proof of Lemma \ref{lemma:rate}.
\end{proof}

\subsection{Proof of Lemma \ref{lemma:UnprofitableDeviation}}\label{sec:ProofLemma42}

We fix $t\in\{1,\ldots,n\}$. This proof relies on the choice of the auxiliary random variable $W_{t} = (M, U^{t-1})$ in \eqref{eq:DefAuxRV}. According to \eqref{eq:RateW}, the rate constraint $I_{\QQ}(U_t;W_t)\leq \mathsf{R}_t$ is satisfied. We consider a single-letter deviation for the encoder $\PP_{W_t|U_t}\in \mathbb{D}(\QQ^{\sigma}_{W_t},\mathsf{R}_t)$, and we construct a new distribution
\begin{align}
\PP_{U^nW^n} =  \QQ^{\sigma}_{U^{t-1}U_{t+1}^nW^{t-1}W_{t+1}^n}
\otimes \PP_{U_t} \PP_{W_t|U_t}
.\label{eq:distributionDev0}
\end{align}
In this construction $(U_t,W_t)$ are independent of $(U^{t-1},U_{t+1}^n,W^{t-1},W_{t+1}^n)$, but the marginal distributions of $U^n$ and of each of the $W_{t'}$ with $t'\in\{1,\ldots,n\}$ are preserved, i.e. $\PP_{U^n} = \PP_{U}^{\otimes n}$ and $\PP_{W_{t'}} = \QQ^{\sigma}_{W_{t'}}$ for all $t'\in\{1,\ldots,n\}$. In fact, only the correlation between $U_t$ and $W_t$ has changed. We define the strategy that assign for all $(u^n,m)$
\begin{align}
\widetilde{\sigma}\big(m\big|u^n\big) =& \frac{\QQ^{\sigma}_{M}(m)}{\PP^{\otimes n}_U(u^n)} \prod_{t'=1}^n\PP_{U_{t'}|MU^{t'-1}}(u_{t'}|m,u^{t'-1}),\label{eq:SigmaTilde0}
\end{align}
where the distributions $\QQ^{\sigma}_{M}=\QQ^{\sigma}_{W_1}$, and $\PP_{U_{t'}|MU^{t'-1}} = \PP_{U_{t'}|W_{t'}}$ for all $t'\in\{1,\ldots,n\}$ are given by $\PP_{U^nW^n} $. 

The first step is to show that the strategy $\widetilde{\sigma}$ is well defined. By construction, the distribution $\PP_{U^nW^n}$ of \eqref{eq:distributionDev0} has a marginals $\PP_{U}^{\otimes n}$ and $\QQ^{\sigma}_{W_1}$. 
The chain rule ensures that
\begin{align}
&\PP_{U}^{\otimes n}(u^n) = \sum_{w_1}\PP_{U^nW_1}(u^n,w_1)\nonumber \\
=& \sum_{w_1}\QQ^{\sigma}_{W_1}(w_1) \prod_{t'=1}^n\PP_{U_{t'}|W_1U^{t'-1}}(u_{t'}|w_1,u^{t'-1}) \nonumber \\
\Longleftrightarrow  &\sum_{m}\QQ^{\sigma}_{M}(m)\prod_{t'=1}^n \PP_{U_{t'}|MU^{t'-1}}(u_{t'}|m,u^{t'-1}) =  \PP_{U}^{\otimes n}(u^n).\label{eq:cond3VT7}
\end{align}
By using \eqref{eq:cond3VT7}, we have for all $u^n\in\mc{U}^n$,
\begin{align}
&\sum_{m}\widetilde{\sigma}\big(m\big|u^n\big) \nonumber\\
=& \sum_{m} \frac{\QQ^{\sigma}_{M}(m)}{\PP^{\otimes n}_U(u^n)} \prod_{t'=1}^n \PP_{U_{t'}|MU^{t'-1}}(u_{t'}|m,u^{t'-1}) =1.
\end{align}
This ensures that $\widetilde{\sigma}$ is well defined.

The second step is to show that the strategies $(\widetilde{\sigma},\tau)$ and the distribution $\PP_{U^nW^n}$ lead to the same expected cost.
\begin{align}
&\sum_{u^n,w^n}\PP_{U^nW^n}(u^n,w^n) \nonumber\\
&\quad \times\bigg(\frac1n \sum_{t'=1}^n \sum_{v_{t'}}\QQ^{\tau}_{V_{t'}|W_{t'}}(v_{t'}|w_{t'})c_{\mathsf{e}}(u_{t'},v_{t'})\bigg)\nonumber\\
=&\sum_{u^n,m}\PP_{U^nM}(u^n,m) \nonumber\\
&\quad\times\bigg(\frac1n \sum_{t'=1}^n \sum_{v_{t'}}\QQ^{\tau}_{V_{t'}|MU^{t'-1}}(v_{t'}|m,u^{t'-1})c_{\mathsf{e}}(u_{t'},v_{t'})\bigg)\label{eq:ConvEqual01}\\
=&\sum_{u^n,m}\PP_{U^nM}(u^n,m) \bigg(\frac1n \sum_{t'=1}^n \sum_{v^n}\tau(v^n|m) c_{\mathsf{e}}(u_{t'},v_{t'})\bigg)\label{eq:ConvEqual011}\\
=&\sum_{u^n,m}\QQ^{\sigma}_{M}(m)  \prod_{t'=1}^n\PP_{U_{t'}|MU^{t'-1}}(u_{t'}|m,u^{t'-1}) \nonumber\\
&\quad\times\sum_{v^n} \tau(v^n|m) \bigg(\frac1n \sum_{t'=1}^n c_{\mathsf{e}}(u_{t'},v_{t'})\bigg)\label{eq:ConvEqual03}\\
=&\sum_{u^n,m,v^n}\PP_U^{\otimes n}(u^n)\widetilde{\sigma}(m|u^n)\tau(v^n|m) \bigg(\frac1n \sum_{t'=1}^n c_{\mathsf{e}}(u_{t'},v_{t'})\bigg).\label{eq:ConvEqual04}
\end{align}
Equation \eqref{eq:ConvEqual01} comes from the identification of the axiliary random variables $W_{t'} = (M, U^{t'-1})$, for all $t'\in\{1,\ldots,n\}$, that implies that $(U^n,W^n)=(U^n,M,U^1,\ldots,M,U^{n-1})$ induce the same expected cost as $(U^n,M)$. Equation \eqref{eq:ConvEqual011} comes from the Markov chain $U^{t'-1}  -\!\!\!\!\minuso\!\!\!\!- M  -\!\!\!\!\minuso\!\!\!\!- V^n$. 
Equation \eqref{eq:ConvEqual03} comes from \eqref{eq:cond3VT7} and the hypothesis that the marginal distribution $\PP_{W_1} = \QQ^{\sigma}_{W_1} = \QQ^{\sigma}_{M}$ is preserved. 
Equation \eqref{eq:ConvEqual04} comes from the definition of the strategy $\widetilde{\sigma}$.

However, $\sigma\in \mathsf{BR}(\tau)$, therefore \eqref{eq:ConvEqual01}-\eqref{eq:ConvEqual04} imply that
\begin{align}
&  \frac1n \sum_{t'=1}^n \bigg( \sum_{u_{t'},w_{t'},v_{t'}}\PP_{U}(u_{t'}) \PP_{W_{t'}|U_{t'}}(w_{t'}|u_{t'}) \nonumber\\
&\quad \times \QQ^{\tau}_{V_{t'}|W_{t'}}(v_{t'}|w_{t'})c_{\mathsf{e}}(u_{t'},v_{t'})\bigg)\label{eq:converseIneq22}\\
=&  \sum_{u^n,w^n}\PP_{U^nW^n}(u^n,w^n)\nonumber\\
&\quad \times \bigg(\frac1n \sum_{t'=1}^n \sum_{v_{t'}}\QQ^{\tau}_{V_{t'}|W_{t'}}(v_{t'}|w_{t'})c_{\mathsf{e}}(u_{t'},v_{t'})\bigg)\\
=& \sum_{u^n,m,v^n}\PP_U^{\otimes n}(u^n)\widetilde{\sigma}(m|u^n)\tau(v^n|m) \bigg(\frac1n \sum_{t'=1}^n c_{\mathsf{e}}(u_{t'},v_{t'})\bigg) \\
\leq & \sum_{u^n,m,v^n}\PP_U^{\otimes n}(u^n){\sigma}(m|u^n)\tau(v^n|m) \bigg(\frac1n \sum_{t'=1}^n c_{\mathsf{e}}(u_{t'},v_{t'})\bigg) \label{eq:converseIneq24}\\
=&  \sum_{u^n,w^n}\QQ^{\sigma}_{U^nW^n}(u^n,w^n)\nonumber\\
&\quad \times \bigg(\frac1n \sum_{t'=1}^n \sum_{v_{t'}}\QQ^{\tau}_{V_{t'}|W_{t'}}(v_{t'}|w_{t'})c_{\mathsf{e}}(u_{t'},v_{t'})\bigg)\\
=&  \frac1n \sum_{t'=1}^n \bigg( \sum_{u_{t'},w_{t'},v_{t'}}\PP_{U}(u_{t'}) \QQ^{\sigma}_{W_{t'}|U_{t'}}(w_{t'}|u_{t'}) \nonumber\\
&\quad \times\QQ^{\tau}_{V_{t'}|W_{t'}}(v_{t'}|w_{t'})c_{\mathsf{e}}(u_{t'},v_{t'})\bigg).\label{eq:converseIneq26}
\end{align}
By construction in \eqref{eq:distributionDev0}, $\PP_{U_{t'}W_{t'}} = \QQ^{\sigma}_{U_{t'}W_{t'}}$ for all $t'\neq t$.
Therefore, equations \eqref{eq:converseIneq22}-\eqref{eq:converseIneq26} imply 
\begin{align}
&\sum_{u_t,w_t,v_t}\PP_{U}(u_t) \PP_{W_t|U_t}(w_t|u_t)  \QQ^{\tau}_{V_t|W_t}(v_t|w_t)c_{\mathsf{e}}(u_t,v_t) \nonumber\\
\leq&  \sum_{u_t,w_t,v_t} \PP_{U}(u_t) \QQ^{\sigma}_{W_t|U_t}(w_t|u_t) \QQ^{\tau}_{V_t|W_t}(v_t|w_t)c_{\mathsf{e}}(u_t,v_t), \nonumber
\end{align}
which show that $\QQ^{\sigma}_{W_t|U_t}\in \mathbb{P}(\QQ^{\tau}_{W_tV_t},\mathsf{R}_t)$. 

By definition, $\sigma \in \textsf{BR}_{\mathsf{e}}(\tau)$  achieves the maximum in \eqref{eq:MaxTie}. We use similar arguments to show that
\begin{align}
\E_{\QQ^{\sigma,\tau}}\Big[ c_{\mathsf{d}}(U_t,V_t) \Big] = \!\!\! \max_{\QQ_{W|U}\in \mathbb{P}(\QQ^{\tau}_{W_tV_t},\mathsf{R}_t)}\E\Big[ c_{\mathsf{d}}(U_t,V_t) \Big] .
\end{align}
This concludes the proof of Lemma \ref{lemma:UnprofitableDeviation}.

\subsection{Cardinality bounds for the auxiliary random variables}

According to \cite[Corollary~17.1.5, pp.~157]{rockafellar1970convex}, the convexification 
of the function $C$ can be obtained by minimizing over convex combinations of at most $|\mc{T}|=2$ points.

\begin{lemma}\label{lemma:CardT}[from Corollary~17.1.5, pp.~157 in \cite{rockafellar1970convex}]
\begin{align}
&\min_{(\tilde{\alpha}_t)_t\geq0,\atop (\tilde{\mathsf{R}}_t)_t\geq0}\bigg\{\sum_t {\tilde{\alpha}}_t  C(\tilde{\mathsf{R}}_t)\;\; \bigg|\;\;  \sum_{t}\tilde{\alpha}_t = 1,  \sum_{t}\tilde{\alpha}_t \tilde{\mathsf{R}}_t \leq \mathsf{R}\bigg\}\nonumber \\
=& \min_{\alpha\in[0,1], \mathsf{R}_1\geq0, \mathsf{R}_2\geq0,\atop \alpha \mathsf{R}_1 +(1- \alpha) \mathsf{R}_2 \leq \mathsf{R}} \!\!\alpha C(\mathsf{R}_1) + (1- \alpha)C(\mathsf{R}_2) .
\end{align}
\end{lemma}

By building on Caratheodory's Theorem, it is possible to reduce the cardinality of $W$ while preserving: 1. the optimality property of the encoder, and 2. the tie-breaking rule.

\begin{lemma}\label{lemma:CardW}
Given $\mathsf{R}$, 
for all distribution $\QQ_{\widetilde{W}V}$, there exists $\QQ_{{W}V}$ with $|\mc{W}|= |\mc{U}|+3$ such that 
\begin{align}
\min_{\QQ_{\widetilde{W}|U}\in \mathbb{D}(\QQ_{\widetilde{W}},\mathsf{R})}\; \E\Big[ c_{\mathsf{e}}(U,V) \Big]=&\min_{\QQ_{W|U}\in \mathbb{D}(\QQ_{W},\mathsf{R})}\; \E\Big[ c_{\mathsf{e}}(U,V) \Big],\\
\max_{\QQ_{\widetilde{W}|U}\in \mathbb{P}(\QQ_{\widetilde{W}V},\mathsf{R})}\; \E\Big[ c_{\mathsf{d}}(U,V) \Big]=&\max_{\QQ_{W|U}\in \mathbb{P}(\QQ_{WV},\mathsf{R})}\; \E\Big[ c_{\mathsf{d}}(U,V) \Big].
\end{align}
\end{lemma}


\subsection{Proof of Lemma \ref{lemma:CardW}}\label{sec:proofLemmaCardW}
\subsubsection{Convex Optimization Problem}
We consider the set 
\begin{align}
&\mathbb{P}(\QQ_{WV},\mathsf{R})  \nonumber\\
=&  \underset{\PP_{W|U}}{\argmin}  \bigg\{ \E[c_{\mathsf{e}}(U,V) ] \bigg| I_{\PP}(U;W)   \leq  \mathsf{R}, \PP_{W} = \QQ_{W} \bigg\}.\label{eq:OptimizationPb}
\end{align}
All convex combinations of distributions $\PP_{W|U}^1$ and $\PP_{W|U}^2$ in $\mathbb{P}(\QQ_{WV},\mathsf{R}) $ have the same marginal $\QQ_W$,  the same values for $\E[c_{\mathsf{e}}(U,V) ]$, and satisfy the information constraint $I_{\PP}(U;W)   \leq  \mathsf{R}$. Therefore $\mathbb{P}(\QQ_{WV},\mathsf{R}) $ is a convex set.

We define the function $h:\R \to (-\infty,+\infty]$ by:
\begin{align}
h : (x) \mapsto 
\begin{cases}
x \log_2(x)&\text{ if }x >0,\\
0&\text{ if }x =0,\\
+\infty&\text{ otherwise. }
\end{cases}
\end{align}
We introduce the notation $\lambda = \QQ_W\in\Delta(\mc{W})$, the cost $c^{\mathsf{e}}_{uw} = \sum_{v}\QQ_{V|W}(v|w) c_{\mathsf{e}}(u,v)$ and the parameter $p\in \R^{|\mc{U}\times\mc{W}|}$ where $p_{uw}=\PP_{U|W}(u|w)$. Without loss of generality, we assume that $\mc{W} = \supp \lambda$. We reformulate the optimization problem \eqref{eq:OptimizationPb}.
\begin{align}
 \min_{p} \;\;\, &\sum_w \lambda_w\;\sum_u p_{uw}\;c^{\mathsf{e}}_{uw} \label{eq:OptimizationPb2}\\
\;\;\text{s.t.} \quad  &\sum_w \lambda_w \;p_{uw} - \PP_{U}(u) = 0, \qquad \forall u\in\mc{U},\label{eq:cons1}\\
\phantom{\;\;\text{s.t.}} \quad  &\sum_u \lambda_w\;p_{uw} - \lambda_w =0, \qquad\quad\; \forall w\in\mc{W},\label{eq:cons2}\\
\phantom{\;\;\text{s.t.}} \quad  &\sum_w \lambda_w \sum_u h(p_{uw})+H(U) - \mathsf{R} \leq 0.\label{eq:cons3}
\end{align}
The function $h$ is proper lower semicontinuous convex, thus \eqref{eq:OptimizationPb2} is an ordinary convex optimization problem \cite[pp.~273]{rockafellar1970convex}.

Lagrangian has multipliers $\nu_1\in\R^{|\mc{U}|}$, $\nu_2\in\R^{|\mc{W}|}$, $\nu_3\geq0$. 
\begin{align}
&\mc{L}^{\lambda}(p,\nu_1,\nu_2,\nu_3)\nonumber\\
=& \sum_w \lambda_w\;\sum_u p_{uw} \; c^{\mathsf{e}}_{uw}  +  \sum_u \nu_{1,u} \Big(\sum_w \lambda_w \; p_{uw} - \PP_U(u) \Big) \nonumber\\
&+ \sum_w\nu_{2,w} \Big(\sum_u \lambda_w  \; p_{uw} - \lambda_w  \Big) \nonumber\\
&+  \nu_3 \Big(\sum_w \lambda_w \sum_u h(p_{uw})+H(U) - \mathsf{R}\Big)\\
=& \sum_{u,w} \lambda_w \Big( p_{uw}  \big(c^{\mathsf{e}}_{uw}  +   \nu_{1,u}   + \nu_{2,w}\big)  +  \nu_3 h(p_{uw}) \nonumber\\
& - \PP_U(u) (\nu_{1,u} + \nu_{2,w}) + \nu_3 \big(  H(U) - \mathsf{R} \big) \Big).\label{eq:Lagrangian1}
\end{align}

We introduce the subdifferential $\partial f(p) $ of a function $f$, see \cite[pp.~215]{rockafellar1970convex}. The first-order conditions writes
\begin{align}
0 \in& \partial\mc{L}^{\lambda}(p,\nu_1,\nu_2,\nu_3),\nonumber\\
\Longleftrightarrow\;\; 0 \in&  \sum_{u,w} \lambda_w \partial \Big( p_{uw}  \big(c^{\mathsf{e}}_{uw}  +   \nu_{1,u}   + \nu_{2,w}\big)  +  \nu_3 h(p_{uw})\Big),\label{eq:subDiff_2}\\
\Longleftrightarrow\;\; 0 \in& \lambda_w \Big( c^{\mathsf{e}}_{uw} + \nu_{1,u}+ \nu_{2,w} +\nu_3\; \partial h\big(p_{uw}\big)\Big),\;\;\; \forall (u,w),\label{eq:subDiff3}
\end{align}
where  \eqref{eq:subDiff_2} comes from \cite[Theorem~23.8, pp.~223]{rockafellar1970convex}, and \eqref{eq:subDiff3} comes from the fact that the function $f_{uw} : p \mapsto p_{uw}  \big(c^{\mathsf{e}}_{uw}  +   \nu_{1,u}   + \nu_{2,w}\big)  +  \nu_3 h(p_{uw})$ is constant with respect to the other parameters $p_{s'w'}$ with $(s',w')\neq (u,w)$.

Suppose that $p^{\star}$ is a solution of \eqref{eq:OptimizationPb2}. By \cite[Theorem~28.3, pp.~281]{rockafellar1970convex}, there exists multipliers such that $(p^{\star}, \nu^{\star}_1,\nu^{\star}_2,\nu^{\star}_3)$ is a saddle-point of the Lagrangian, i.e. for all $(p,\nu_1,\nu_2,\nu_3)$,
\begin{align}
&\mc{L}^{\lambda}(p^{\star}, \nu_1,\nu_2,\nu_3)
 \leq \mc{L}^{\lambda}(p^{\star}, \nu^{\star}_1,\nu^{\star}_2,\nu^{\star}_3) \leq \mc{L}^{\lambda}(p, \nu^{\star}_1,\nu^{\star}_2,\nu^{\star}_3).\nonumber
\end{align}
$(p^{\star}, \nu^{\star}_1,\nu^{\star}_2,\nu^{\star}_3)$ is a saddle-point if and only if it satisfies the first-order conditions \eqref{eq:subDiff3} 
and Karush-Kuhn-Tucker conditions
\begin{align}
&\sum_w \lambda_w \;p^{\star}_{uw} - \PP_{U}(u) = 0, \qquad \forall u\in\mc{U},\label{eq:KKT_1}\\
&\sum_u \lambda_w\;p^{\star}_{uw} - \lambda_w =0, \qquad\quad\;\; \forall w\in\mc{W},\label{eq:KKT_2}\\
&\sum_w \lambda_w \sum_u h(p^{\star}_{uw})+H(U) - \mathsf{R} \leq 0,\label{eq:KKT_3}\\
&\nu^{\star}_3 \bigg(\sum_w \lambda_w \sum_u h(p^{\star}_{uw})+H(U) - \mathsf{R}\bigg) = 0.\label{eq:KKT_4}
\end{align}

\subsubsection{Caratheodory reduction}\label{sec:Caratheodory}

We denote by $p^{\star}$ a solution of \eqref{eq:OptimizationPb2} and we introduce the notation $c^{\mathsf{d}}_{uw} = \sum_{v}\QQ_{V|W}(v|w) c_{\mathsf{d}}(u,v)$. We define the family $F= (x_w)_{w\in\mc{W}}$,
\begin{align}
\!\!\!x_w = \Big(p^{\star}_{uw}, \sum_u h(p^{\star}_{uw}), \sum_u p^{\star}_{uw} \; c^{\mathsf{e}}_{uw},  \sum_u p^{\star}_{uw} \; c^{\mathsf{d}}_{uw},1\Big).\label{eq:FamilyX}
\end{align}
The family $F$ belongs to $\Delta(\mc{U})\times \R^4$ of dimension $|\mc{U}|+3$. 

Suppose that $|\mc{W}|>|\mc{U}|+3$, then the family $F$ has more elements than the dimension of the vector space $\R^{|\mc{U}|+3}$. Since it is composed of linearly dependent vectors, there exists $\mu\in\R^{|\mc{W}|}$ such that  $\mu\neq 0$ and 
\begin{align}
0 = \sum_w \mu_w \Big(p^{\star}_{uw}, \sum_u h(p^{\star}_{uw}), \sum_u p^{\star}_{uw} \; c^{\mathsf{e}}_{uw},  \sum_u p^{\star}_{uw} \; c^{\mathsf{d}}_{uw},1\Big) .\label{eq:MuZero}
\end{align}

Inspired by the proof of Caratheodory's Theorem in \cite[Theorem~17.1, pp.155]{rockafellar1970convex}, we take $\gamma\in \R$ such that $\gamma \mu \leq \lambda$. We denote by $\mc{W}^+$ the subset of $w\in\mc{W}$ such that $\mu_w>0$ and $\mc{W}^-$ the subset of $w\in\mc{W}$ such that $\mu_w<0$. Therefore, 
\begin{align}
\gamma \in \bigg[\;\max_{w\in\mc{W}^-} \frac{\lambda_w}{\mu_w}\; , \; \; \min_{w\in\mc{W}^+} \frac{\lambda_w}{\mu_w}\; \;\bigg].\label{eq:IntervalGamma}
\end{align}
We define $\widetilde{\lambda} = \lambda - \gamma\mu$ which is positive due to $\gamma \mu \leq \lambda$. The last component in \eqref{eq:MuZero} ensures $\sum_w \mu_w =0$, thus $\sum_w\widetilde{\lambda}_w = \sum_w\lambda_w=1$, which implies that $\widetilde{\lambda} \in\Delta(\mc{W})$. Note that $\supp \widetilde{\lambda} \subset \supp \lambda$ and this inclusion is strict when $\gamma$ is at the boundary of the interval in \eqref{eq:IntervalGamma}.

We reformulate the optimization problem \eqref{eq:OptimizationPb2} with $\widetilde{\lambda}$. 
\begin{align}
 \min_{p} \;\;\, &\sum_w \widetilde{\lambda}_w\;\sum_u p_{uw}\;c^{\mathsf{e}}_{uw} \label{eq:OptimizationPb3}\\
\;\;\text{s.t.} \quad  &\sum_w \widetilde{\lambda}_w \;p_{uw} - \PP_{U}(u) = 0, \qquad \forall u\in\mc{U},\\
\phantom{\;\;\text{s.t.}} \quad  &\sum_u \widetilde{\lambda}_w\;p_{uw} - \widetilde{\lambda}_w =0, \qquad\quad\;\; \forall w\in\mc{W},\\
\phantom{\;\;\text{s.t.}} \quad  &\sum_w \widetilde{\lambda}_w \sum_u h(p_{uw})+H(U) - \mathsf{R} \leq 0.
\end{align}

We show that the solution $(p^{\star}, \nu^{\star}_1,\nu^{\star}_2,\nu^{\star}_3)$ that satisfy \eqref{eq:subDiff3} and \eqref{eq:KKT_1}-\eqref{eq:KKT_4} is also a saddle-point of the Lagrangian $\mc{L}^{\tilde{\lambda}}(p,\nu_1,\nu_2,\nu_3)$ of problem \eqref{eq:OptimizationPb3}. 
For all $(u,w)\in\mc{U}\times \mc{W}$, first-order conditions  are satisfied
\begin{align}
0 \in& \lambda_w \Big( c^{\mathsf{e}}_{uw} + \nu^{\star}_{1,u}+ \nu^{\star}_{2,w} +\nu^{\star}_3\; \partial h\big(p^{\star}_{uw}\big)\Big),\label{eq:subDiff30}\\
\Longrightarrow\qquad 0 \in& \widetilde{\lambda}_w \Big( c^{\mathsf{e}}_{uw} + \nu^{\star}_{1,u}+ \nu^{\star}_{2,w} +\nu^{\star}_3\; \partial h\big(p^{\star}_{uw}\big)\Big)\label{eq:subDiff40}.
\end{align}
By definition $\widetilde{\lambda} = \lambda - \gamma\mu$ and $\mu$ satisfies \eqref{eq:MuZero}. Thus the Karush-Kuhn-Tucker conditions are also satisfied with respect to $\widetilde{\lambda}$.
\begin{align}
&\sum_w \widetilde{\lambda}_w \;p^{\star}_{uw} - \PP_{U}(u) = 0, \qquad \forall u\in\mc{U},\label{eq:KKT_1}\\
&\sum_u \widetilde{\lambda}_w\;p^{\star}_{uw} - \widetilde{\lambda}_w =0, \qquad\quad\;\; \forall w\in\mc{W},\label{eq:KKT_2}\\
&\sum_w \widetilde{\lambda}_w \sum_u h(p^{\star}_{uw})+H(U) - \mathsf{R} \leq 0,\label{eq:KKT_3}\\
&\nu^{\star}_3 \bigg(\sum_w \widetilde{\lambda}_w \sum_u h(p^{\star}_{uw})+H(U) - \mathsf{R}\bigg) = 0.\label{eq:KKT_4}
\end{align}
By \cite[Theorem~28.3, pp.~281]{rockafellar1970convex}, \eqref{eq:subDiff40} and \eqref{eq:KKT_1}-\eqref{eq:KKT_4} implies that $(p^{\star}, \nu^{\star}_1,\nu^{\star}_2,\nu^{\star}_3)$ is a saddle-point of the Lagrangian $\mc{L}^{\widetilde{\lambda}}(p,\nu_1,\nu_2,\nu_3)$.  Thus $p^{\star}$ is an optimal solution to both problems \eqref{eq:OptimizationPb2} and \eqref{eq:OptimizationPb3}. 
Moreover, optimal values also coincide 
\begin{align}
\sum_w \widetilde{\lambda}_w\;\sum_u p^{\star}_{uw} \;c^{\mathsf{e}}_{uw} =  \sum_w \lambda_w \sum_u p^{\star}_{uw} \;c^{\mathsf{e}}_{uw} = C_{\mathsf{e}}^{\star}.
\end{align}


\subsubsection{Tie-breaking rule}

The pessimistic tie-breaking rule of \eqref{eq:DefinitionFunctionC} reformulates as a convex optimization problem.
\begin{align}
 \max_{p} \;\;\, &\sum_w \lambda_w\;\sum_u p_{uw}\;c^{\mathsf{d}}_{uw} \label{eq:OptimizationPb5}\\
\;\;\text{s.t.} \quad  &\sum_w \lambda_w\;\sum_u p_{uw}\;c^{\mathsf{e}}_{uw} - C_{\mathsf{e}}^{\star}\leq0,\\
\;\;\text{and} \quad &p\text{ satisfy \eqref{eq:cons1}, \eqref{eq:cons2}, \eqref{eq:cons3}.}
\end{align}

We use the arguments of Sec.~\ref{sec:Caratheodory} to show that if $p^{\star}$ is a solution to \eqref{eq:OptimizationPb5} with respect to $\lambda$, then it is also a solution with respect to $\tilde{\lambda}$. By taking $\gamma$ at the boundary of the interval \eqref{eq:IntervalGamma}, at least one component of $\tilde{\lambda}$ is zero. By using this argument repeatedly, we show that it is enough to consider $|\mc{W}|=|\mc{U}|+3$. This concludes the proof of Lemma~\ref{lemma:CardW}.



\end{document}